\documentclass[numbers,preprint]{sigplanconf} 
\usepackage{amsthm}
\usepackage{bussproofs}
\usepackage{tikz}
\usepackage{stmaryrd}
\usepackage{mathtools}
\usepackage{cmll}
\usepackage{enumitem}
\usepackage{macros}
\usepackage{bm}

\SetLabelAlign{parright}{\parbox[t]{\labelwidth}{\raggedleft#1}}
\setlist[description]{align=right,style=multiline,leftmargin=2.2cm,labelindent=0.2cm}

\newcommand{\balanced}[1]{\textnormal{bal}(#1)}
\newcommand{\measuredunitint}{\mathopen{\bm{[}}\measured{0,1}\mathclose{\bm{]}}}
\newcommand{\realposext}{\mathbf{\bar{R}_{+}}}
\newcommand{\nat}{\mathtt{nat}}
\newcommand{\myparagraph}[1]{\vspace{0.2cm}\noindent\textbf{#1.}~}
\newcommand{\exchange}{\mathtt{xch}}
\newcommand{\der}[1]{\de{der}_{\cond{#1}}}
\newcommand{\dig}[1]{\de{dig}_{\cond{#1}}}

\theoremstyle{theorem}
\newtheorem{theorem}{Theorem}

\theoremstyle{definition}
\newtheorem{definition}{Definition}

\theoremstyle{remark}
\newtheorem*{notation}{Notations}
\newtheorem{example}{Example}
\newtheorem{remark}{Remark}

\usepackage{amsmath,amssymb,hyperref,subfig}

\begin{document}

\setlength{\pdfpageheight}{\paperheight}
\setlength{\pdfpagewidth}{\paperwidth}

%\conferenceinfo{LICS '16}{July 05 - 08, 2016, New York, NY, USA} 
%\copyrightyear{2016} 
%\copyrightdata{978-1-4503-4391-6/16/07}%\acmPrice{\$15.00}
%\copyrightdoi{http://dx.doi.org/10.1145/2933575.2934568}

               % ACM gets exclusive license to publish, 
                                  % you retain copyright

%\CopyrightYear{2016} 
%\setcopyright{acmlicensed}
%\conferenceinfo{LICS '16,}{July 05 - 08, 2016, New York, NY, USA}
%\isbn{978-1-4503-4391-6/16/07}\acmPrice{\$15.00}
%\doi{http://dx.doi.org/10.1145/2933575.2934568}
%
% Uncomment one of the following two, if you are not going for the 
% traditional copyright transfer agreement.

%\permissiontopublish             % ACM gets nonexclusive license to publish
                                  % (paid open-access papers, 
                                  % short abstracts)

%\titlebanner{Interaction graphs: full linear logic}        % These are ignored unless
%\preprintfooter{Preprint submitted to LICS2016}   % 'preprint' option specified.

\title{Interaction Graphs: Full Linear Logic}
%\subtitle{Subtitle Text, if any}

\authorinfo{Thomas Seiller}
           {Department of Computer Science, University of Copenhagen, Njalsgade 128, bygning 24, 2300 K\o benhavn S, Denmark}
           {seiller@di.ku.dk}
%\authorinfo{Name2\and Name3}
  %         {Affiliation2/3}
     %      {Email2/3}

\maketitle

\begin{abstract}
Interaction graphs were introduced as a general, uniform, construction of dynamic models of linear logic, encompassing all \emph{Geometry of Interaction} (GoI) constructions introduced so far. This series of work was inspired from Girard's hyperfinite GoI, and develops a quantitative approach that should be understood as a dynamic version of weighted relational models. Until now, the interaction graphs framework has been shown to deal with exponentials for the constrained system ELL (Elementary Linear Logic) while keeping its quantitative aspect. Adapting older constructions by Girard, one can clearly define \enquote{full} exponentials, but at the cost of these quantitative features. We show here that allowing interpretations of proofs to use continuous (yet finite in a measure-theoretic sense) sets of states, as opposed to earlier Interaction Graphs constructions were these sets of states were discrete (and finite), provides a model for full linear logic with second order quantification.
\end{abstract}


% general terms are not compulsory anymore, 
% you may leave them out

\newline

\section{Introduction}

This work deals with so-called dynamic models of proof theory, such as game semantics and geometry of interaction. It extends previous work providing a uniform construction of quantitative dynamic models of (fragments of) linear logic to full linear logic with second-order quantification.

\myparagraph{Geometry of Interaction}
A Geometry of Interaction (GoI) construction, i.e. a construction that fulfills the GoI research program \cite{towards}, is in a first approximation a representation of linear logic proofs that accounts for the dynamics of cut-elimination. %A proof is no longer a morphism from $A$ to $B$ — a function from $A$ into $B$ — but an operator acting on the space $A\oplus B$. As a consequence, the modus ponens is no longer represented as composition. The operation representing cut-elimination, i.e. the obtention of a cut-free proof of $B$ from a cut-free proof of $A$ and a cut-free proof of $A\multimap B$, consists in constructing the solution -- in general through the so-called execution formula $\Ex(\cdot)$ -- to an equation called the \emph{feedback equation}. 
Contrarily to denotational semantics, a proof $\pi$ and its normalised form $\pi'$ are not represented by the same object, but they remain related through a semantic interpretation of the cut-elimination called the \emph{execution} $\Ex$: $\Ex(\pi)=\pi'$. A GoI construction hence represents both the proofs and their normalisation; it is in some ways an untyped variant of game semantics  \cite{hylandong}.

The further aim of geometry of interaction is to reconstruct logical operations from such a dynamic representation of proofs. The objects of study in a GoI construction are actually a generalisation of the notion of proof -- sometimes called paraproofs%, in the same sense proof structures (i.e. not necessarily correct proof nets) are a generalisation of sequent calculus proofs
. This point of view allows a reconstruction of logic as a description of how these objects interact in the same spirit as realisability \cite{kleene,krivine1,krivine2}: a program is of type $\nat\rightarrow\nat$ because it produces a natural number when given a natural number as an argument. As in game semantics and classical realisability, one can however describe a necessary condition for being the interpretation of a proof, and defines \emph{winning paraproofs} as those objects satisfying it.%A notion of \emph{winning paraproof} -- a winning strategies in game semantics -- is then \note{le bon term} to characterise those 

In spite of their seemingly deep abstraction, the GoI constructions provide mathematical models which are very close to actual computing. As an illustration of this fact, let us mention the \emph{Geometry of Synthesis} program initiated by Ghica \cite{GoS1,GoS2,GoS3,GoS4}. This research program, inspired by geometry of interaction, aims at obtaining logical synthesis methods for VLSI designs.

\myparagraph{Quantitative Semantics}
Quantitative semantics find its origins in Girard's work on functor models for lambda-calculus \cite{FoncteurAnalytiques}. This work, which predates its seminal work on linear logic \cite{linearlogic} and actually inspired it, exhibits for the first time a decomposition of the semantic interpretation of lambda-terms as Taylor series. These series capture a number of information about the time, space, resource consumption of the programs it represents. Quantitative semantics are therefore more involved than so-called \emph{qualitative} semantics, since they mirror more information about the programs that are interpreted. %\note{Even though quantitative techniques have been used  since long in verification, the semantic community's interest is quite recent.}
Recently, quantitative semantics has been used to give denotational semantics for various algebraic extensions of lambda calculus such as probabilistic \cite{probcoh} or differential lambda calculi \cite{finitenessspaces}. Work by Laird, Manzonetto, McCusker and Pagani on weighted relational models \cite{quantdenot} provides a uniform account of several denotational models accounting for quantitative notions, using a refinement of the relational model. 

\myparagraph{Interaction Graphs}
Interaction graphs were first introduced by the author \cite{seiller-goim} as a combinatorial approach to Girard's hyperfinite Geometry of Interaction \cite{goi5}, restricted to the multiplicative fragment of linear logic. An extension capable to deal with additive connectives was then defined \cite{seiller-goiadd} and shown to abstract not only the (additive fragment of the) hyperfinite GoI model but all previously introduced GoI constructions as well. Both papers proposed a model construction in the spirit of Girard's GoI construction where proofs were interpreted by graphs instead of infinite operators. Dealing with exponentials however needs one to consider infinite objects. This is why a third paper \cite{seiller-goig} showed how the construction on graphs can be applied when working with a generalisation of graphs named \emph{graphings}. Graphings are in some sense \emph{geometric realisations} of graphs on a measure space $\measured{X}$ which were first introduced in the context of ergodic theory \cite{adams,gaboriaucost}. This allows not only to consider infinite graphs (which can be used to define exponentials in the same way as the original GoI constructions), but also graphs acting on continuous, thus infinite but finite-measure spaces. This general construction on graphings was shown \cite{seiller-goig} to improve on Girard's hyperfinite GoI \cite{goi5} since it allows a satisfactory treatment of second-order quantification. Lastly, a fourth paper\footnote{Although the author's PhD thesis \cite{seiller-phd} did not contain the general treatment of graphings \cite{seiller-goig}, it already introduced the model of \ELL \cite{seiller-goie} and the restricted theory of graphings this model uses.} \cite{seiller-goie} showed how the consideration of graphings can be used to define \enquote{quantitative} exponential connectives for Elementary Linear Logic \cite{LLL}, a fragment of linear logic that captures elementary time computation \cite{danosjoinet}.

\myparagraph{Unbounded Exponentials and Quantitative Aspects}
The author's work on Interaction Graphs should be understood \cite{seiller-fock} as a \emph{dynamic counterpart} of weighed relational models by Laird \emph{et al.} \cite{quantdenot}, i.e. its relation to standard dynamic models (geometry of interaction, game semantics) is comparable to weighted relational models' relations with standard relational models. Indeed, it provides a uniform construction of models which not only captures all of Girard's GoI models, but also extends them: while Girard's constructions can be understood as interpreting proofs as graphs\footnote{Girard interprets proofs as partial isometries acting on a Hilbert space $\hil{H}$ which, by considering the right basis for $\hil{H}$ correspond to graphs.}, we here interpret proofs as \emph{weighted graphs}, i.e. graphs with weighted edges\footnote{Actually, the most general models are built around the lesser known notion of weighted \emph{graphing}. However, thinking about graphings as graphs should provide the reader with the right intuitions.}. Furthermore, interaction graphs models can reflect these quantitative information at the level of types since the latter are built from an \emph{orthogonality relation} which can take those weights into account. Indeed, the orthogonality relation is defined through a measurement of cycles \cite{seiller-goig} by means of an integral over a finite-measure space -- the \emph{support} of the cycle. In the simplest cases one measures a cycle $\pi$ of suppport $\supp{\pi}$ and weight $\omega(\pi)\in\Omega$, along a measurable map $m:\Omega\rightarrow\realposext$, by the following integral:
\begin{equation}
\int_{\supp{\pi}} m(\omega(\pi))\label{integral}
\end{equation}

Since Interaction Graphs provide a generalisation of Girard's constructions, one could easily adapt the interpretation of exponential connectives from Girard's first constructions \cite{goi1,goi3} to obtain a model of full linear logic. This adaptation would extend to Danos' interpretation of pure lambda-calculus in GoI \cite{danos-phd}. However, this interpretation of exponential connectives corresponds to defining $\oc a$ as a (countable) infinite family of copies of $a$. Thus, even if $a$ is represented by a graphing acting on a space of finite measure, its exponentiated version $\oc a$ acts on a space of infinite measure. This fact hinders the quantitative aspects of our model since it creates cycles $\pi$ whose support $\supp{\pi}$ are spaces of infinite measure. As a consequence, the integral defining the orthogonality relation (Definition \ref{integral}) diverges as soon as the weight is not mapped to $0$, i.e. as soon as $m(\omega(\pi))\neq 0$. The resulting model is therefore no longer capable of depicting quantitative information.

\myparagraph{Contributions}
We define, in the framework of interaction graphs, exponential connectives for full linear logic in a way that preserves the quantitative aspects of the construction, providing the first dynamic model of second-order linear logic accounting for quantitative aspects. Indeed, to the author's knowledge, there exists no game semantics or GoI models for this logical system which include quantitative features. Indeed, although Girard's so-called GoI3 construction \cite{goi3} provides a model of this same logical system, the treatment of exponential connectives in the latter work prevents from any generalisation accounting for quantitative information, as already explained. On the side of game semantics, the quantitative game semantics for linear logic of dal Lago and Laurent \cite{quantgames} does not deal with additives and quantifiers and moreover seems more limited than our own models in the range of quantitative features it can accommodate.

Moreover, we are able to pinpoint the computational principles (represented as measurable maps) that are essential to interpret digging and dereliction, providing new insights on constraint linear logic and their semantics. Formally, this is obtained by exhibiting a single map -- the \emph{exchange} $\exchange$ -- which turns a model of \ELL into a model of \LL. Indeed, as discussed after Definition \ref{microdef}, removing the exchange restricts our model to the model of \ELL based on graphings \cite{seiller-phd,seiller-goie}. 

\myparagraph{Outline of the paper}
We define exponential connectives along the same lines as in our work on \ELL (thus bounded) exponentials \cite{seiller-goie}, avoiding the involvement of infinite-measure sets. With this definition of exponential connectives, one would however expect only a restriction of linear logic, such as \ELL. To bypass this restriction, we relax the notion of states. Indeed, the interpretation of proofs in interaction graphs makes use of so-called \emph{thick graphs} -- or \emph{thick graphings} in the general framework --, which can be understood as graphs with states. While previous work considered only finite sets of states, we loosen this definition to allow for infinite yet finite-measure (actually continuous) sets of states. This modification impacts slightly on the basic notions notions and constructions considered in previous work \cite{seiller-goig}, for which we introduce adequate generalisations. These changes, however, do not raise any technical difficulties. The resulting model is then shown to model digging and dereliction in addition to the principles of Elementary Linear Logic, thus interpreting full linear logic. Finally, we discuss the issue of the representation of cut-elimination in the model.

%This paper is the second GoI-style construction for full linear logic (including additives) and second order quantification. Indeed, Girard's so-called GoI3 construction \cite{goi3} already provided such a model. However, as explained above, the treatment of exponential connectives in the latter work prevents from any generalisation accounting for quantitative information. The main contribution of this paper therefore lies in its quantitative aspects. %Moreover, we are able to pinpoint the computational principles (represented as measurable maps) that are essential to interpret digging and dereliction, i.e. we exhibit a single map -- the \emph{exchange} $\exchange$ -- which turns a model of ELL into a model of LL. L

%\myparagraph{Related Work}

%Apart from geometry of interaction constructions, related work include quantitative realisability \cite{hoffmandallago,brunelquantitative} -- which provides characterizations of computational complexity classes -- and quantitative game semantics for linear logic \cite{quantgames}, although the latter does not deal with additives and quantifiers and seems more limited in its possible quantitative features. %In particular, we believe our construction to be a generalisation of quantitative realisability, allowing for characterisations of a larger family of computational complexity classes \cite{seiller-towards}.

\section{Interaction Graphs}

We start by a discussion meant to give intuitions about the basic principles at work in the interaction graphs models. We illustrate those principles by explaining the notion of \emph{thick and sliced graphs} \cite{seiller-goie}. This discussion is quite informal in that we will only provide explicit and complete definitions of the objects and operations that are essential for the understanding of this paper, to avoid overloading the reader with non-essential definitions. Indeed, the actual model uses thick and sliced \emph{graphings}, a generalisation needed to accomodate both exponentials and quantifiers. Before providing the formal definition of those at the end of the section, we discuss the notion of \enquote{graphs with states} and how it can be generalised to continuous sets of states. 

\subsection{Thick and Sliced Graphs}

The term \enquote{graph} will stand for \enquote{directed weighted graphs}, i.e. directed graphs with a weight function from the set of edges to a monoid\footnote{As we consider paths in the following, the structure of monoid is essential as it allows to define the weight of a path as the product of the weights of the edges that it is composed of.} of weights $\Omega$. Given a graph $G$, we will always denote $E^{G}$ its set of edges, $V^{G}$ its set of vertices, $t^{G}$ and $s^{G}$ its target and source maps, and $\omega^{G}$ its weight map.

The notion of \emph{thick} graphs generalises that of graphs by introducing a set of (control) states -- called a dialect. A graph $G$ with dialect $D^{G}$ is nothing more than a graph whose set of vertices is of the form $V^{G}=S^{G}\times D^{G}$ -- the set $S^{G}$ is called its \emph{support}. The set $D^{G}$ then acts as a set of control states when considering the operation of \emph{execution}, which represents the cut-elimination procedure. When working with (non-thick) graphs, this operation is represented as the computation of a graph of alternating paths; the notion of alternating path between thick graphs gives a particular role to the dialects. Indeed, an alternating path between thick graphs $G_{0}$ and $G_{1}$, with respective dialects $D^{G_{0}}$ and $D^{G_{1}}$, is a finite sequence of edges $e_{0}e_{1}\dots e_{k}$ and a sequence of triples $(s_{i},g^{(0)}_{i},g^{(1)}_{i})_{i=0}^{k+1}$ such that:
\begin{description}[noitemsep,nolistsep]
\item[(Alternation)] $e_{i}\in E^{G}$ if and only if $e_{i+1}\in E^{H}$;
\item[(States)] if $e_{i}\in E^{G_{j}}$ then $\left\{\begin{array}{lll} s^{G_{j}}(e_{i})&=&(s_{i},g^{(j)}_{i})\\ t^{G_{j}}(e_{i})&=&(s_{i+1},g^{(j)}_{i+1})\\ g^{(1-j)}_{i}&=&g^{(1-j)}_{i+1}\end{array}\right.$;
\end{description}
The interpretation of the dialect as a set of (control) states comes from the way its elements are dealt with in the above definition. Now, given two thick graphs $G,H$, the intersection of their supports represent a \emph{cut}; the result of the elimination of this cut is called the \emph{execution} of $G$ and $H$. It is defined as the thick graph $G\plug H$, of support the symmetric difference $S^{G}\Delta S^{H}$, of dialect the product $D^{G}\times D^{H}$, whose edges are exactly the alternating paths between $G$ and $H$ whose source and target lie outside of the cut. This is reminiscent of game semantics' \emph{composition and hiding}: composition corresponds here to the computation of all alternating paths, while hiding corresponds to the restriction to those paths starting and ending outside the cut. The formal definition of execution, in the framework of graphings, is given at the end of this section.

Now, \emph{thick and sliced graphs} further extend this notion of thick graphs by considering finite formal weighted sums $\sum_{i\in I^{G}} \alpha^{G}_{i}G_{i}$ where $I^{G}$ is a finite indexing set, the coefficients $\alpha^{G}_{i}$ are real numbers and $\{G_{i}\}$ is a set of thick graphs sharing the same support (but not sharing the dialects). This notion is crucial for treating additive connectives \cite{seiller-goiadd}. The operation of execution is then extended \enquote{by linearity} (although the sums are not linear combinations), letting:
$$\left(\sum_{i\in I^{G}} \alpha^{G}_{i}G_{i}\right)\plug \left(\sum_{i\in I^{H}} \alpha^{H}_{i}H_{i}\right)=\sum_{(i,j)\in I^{G}\times I^{H}} \alpha^{G}_{i}\alpha^{H}_{j}G_{i}\plug H_{j}$$

\subsection{Continuous Dialects}

Interaction graphs models dealing with exponential connectives of linear logic are based on the notion of \emph{thick and sliced graphings}, obtained by a second layer of generalisation over this notion of thick and sliced graphs \cite{seiller-phd}. While graphings will be introduced formally in the next section (Definition \ref{graphing}), we provide an intuitive description to discuss this generalisation. This discussion can be skipped in a first read, as only the formal definition of microcosm is needed to follow the next section.

Graphings are in some sense \emph{geometric realisations} of graphs on a measure space $\measured{X}$. Specifically, a graphing $G$ is defined as a graph such that for each edge $e\in E^{G}$, $s^{G}(e)$ and $t^{G}(e)$ are measurable subsets of $\measured{X}$, and there is a measurable map $\phi_{e}^{G}: s^{G}(e)\rightarrow t^{G}(e)$ which \emph{realises} $e$. As for graphs, one can define \emph{thick and sliced graphings} by first defining thick graphings -- graphings with a dialect, then consider formal weighted sums of those. It is natural, while working with graphings, to consider dialects themselves as measure spaces, and more precisely (finite) discrete probability spaces. A thick graphing of dialect $\measured{D}$ is then easily described as a graphing over the measure space $\measured{X}\times\measured{D}$.

The purpose of the current work is to extend this definition to allow for \emph{continuous dialects}, i.e. continuous measure spaces in place of discrete ones. We will show how to define in this setting the interpretation of second order linear logic without hindering the \enquote{quantitative} features of the interaction graphs construction. This however comes with a small drawback in the form of a minor complexification of the framework, which we now explain. 

We did not dwell on this point earlier, but thick graphs (graphs with dialects) are considered \emph{up to}\footnote{It is important to remark here that we don't consider the set of graphings quotiented modulo renaming, but we want to be able to formalise this notion of equivalence.} renaming of their dialect; a thick graph $G$ which is a dialect-renaming of a thick graph $F$ is called a \emph{variant} of $F$ (Definition \ref{variant}). To define correctly this notion of variant one needs to consider bijections between the dialects. However, when considering graphings and replacing the dialects with possibly continuous probability spaces, we face a problem when considering the following two probability spaces: $\mathbf{k}=\{1,\dots,k\}$ with discrete measure, and $\measured{[0,1]}$ with Lebesgue measure. Indeed, any thick graphing $G$ with dialect $\measured{k}$ has a variant $H$ with dialect $\measured{[0,1]}$: each element $i\in \measured{k}$ is represented by the interval $I_{i}=[i/(k+1),(i+1)/(k+1)]$, and an edge of source $(v,i)$ and target $(v',j)$ realised by a map $\phi: (v,i)\rightarrow (v',j)$ in $G$ is realised in $H$ by $\phi_{1}\times T_{i,j}$ where $T_{i,j}$ is the translation $x\mapsto x+(j-i)/(k+1)$ and $\phi_{1}$ is the map $v\rightarrow v'$ underlying\footnote{Since $\measured{k}$ is discrete, any measurable map $\phi:(v,i)\rightarrow (v,j)$ is defined from a measurable map $\phi_{1}:v\rightarrow v'$ by $\phi(x,i)=(\phi_{1}(x),j)$.} $\phi$. However, this cannot be formalised through an adequate notion of bijection: here we would expect \emph{Borel isomorphisms} since we work with measure spaces, but no such isomorphism exists between $\measured{k}$ and $\measuredunitint$. To avoid these troubles, we will therefore consider all our dialects to be isomorphic to $\measuredunitint$ with its Lebesgue measure. Since, as we just explained, a graphing with discrete dialect always has a \enquote{variant} with $\measuredunitint$ as dialect, and since thick graphings are considered \emph{up to} renaming, this restriction is seamless.

The second change from earlier work \cite{seiller-goig} is that we need to consider an extension of the notion of \emph{microcosm}. A microcosm $\microcosm{m}$ was defined as a monoid of measurable maps $\measured{X}\rightarrow\measured{X}$ used to consider \enquote{restrictions} of the model to $\microcosm{m}$-graphings: graphings whose realisers -- i.e.\ the maps that realises edges -- are restrictions of maps in $\microcosm{m}$. This original notion of microcosm did not incorporate the dialect. This is explained by the fact that the latter was discrete, and therefore any measurable maps realising an edge in a thick graphing could be described as a product of a measurable maps from $\measured{X}$ to $\measured{X}$ with a partial bijection on the dialect. Now that we allow for continuous dialects, one can consider realisers of edges that do not simply arise in this way from\footnote{As an example, one can consider the \emph{exchange} map defined below (Definition \ref{microdef}) and which is needed to interpret both digging and dereliction.} a map $\measured{X}\rightarrow\measured{X}$. The following definition therefore adapts (in fact extends) the previously considered notion of microcosm in a very natural way in order to incorporate this change. Let us stress that for technical reasons discussed in earlier work \cite{seiller-goig}, the measurable maps considered should be non-singular transformations\footnote{Let $\measured{X}=(X,{\mathcal B},\mu)$ be a measure space. A measurable map $f:X\rightarrow X$ is non-singular when $\forall A\in{\mathcal B}, \mu(f^{-1}(A))=0\Leftrightarrow \mu(A)=0$.} which are measurable-preserving, i.e. map measurable sets to measurable sets.

\begin{definition}[Microcosm]\label{microcosm}
Let $\measured{X}$ be a measure space. A \emph{microcosm} is a monoid (for the composition of functions) of measurable-preserving non-singular transformations $\measured{X}\times\measuredunitint\rightarrow\measured{X}\times\measuredunitint$.
\end{definition}

\subsection{Graphings and Exponential-Free Linear Logic}

This section is meant to recall the main results of previous work \cite{seiller-goig}, to which we refer the reader for a complete picture.
We first define weighted (thick) graphings, a generalisation of the homonymous notion considered by Adams \cite{adams} and later by Gaboriau \cite{gaboriaucost}.

\begin{definition}[Graphing]\label{graphing}
Let $\microcosm{m}$ be a microcosm, $\Omega$ a monoid of weights, $S^{G}$ a measurable subset of $\measured{X}$ and $\measured{D}^{G}$ a probability space isomorphic to $\measuredunitint$. A thick $\Omega$-weighted $\microcosm{m}$-graphing $G$ of support $S^{G}$ and dialect $\measured{D}^{G}$ is given by a set of \emph{edges} $E^{G}$ and $\forall e\in E^{G}$:
\begin{itemize}[noitemsep,nolistsep]
\item a \emph{source} $s^{G}(e)$, i.e. a measurable subset of $S^{G}\times D^{G}$;
\item a \emph{realiser} $\phi^{G}_{e}\in\microcosm{m}$ such that $\phi^{G}(s^{G}(e))\subset S^{G}\times D^{G}$;
\item a \emph{weight} $\omega^{G}(e)$.
\end{itemize}
For all edge $e\in E^{G}$, one can then define the \emph{target} $t^{G}(e)$ of $e$ as the measurable subset $\phi^{G}_{e}(s^{G}(e))$.

A graphing $G$ is \emph{dialect-free} if it does not make use of its dialect, i.e. if for all edge $e$, $\phi_{e}^{G}=\tilde{\phi}_{e}^{G}\times\textnormal{Id}_{\measured{D}^{G}}$, with $\tilde{\phi}_{e}^{G}:\measured{X}\rightarrow\measured{X}$.
\end{definition}

\begin{notation}
Let $A$ be a graphing, $B$ a Borel automorphism of $\measured{X}\times\measuredunitint$. We denote $B(A)$ the graphing whose edges are $B^{-1}\circ \phi\circ B$; up to the automorphism between $\measured{D}^{A}$ and $\measuredunitint$. When $B$ is a Borel automorphism of $\measured{X}$, we abusively denote by $B(A)$ the graphing $B\times\textnormal{Id}_{\measuredunitint}(A)$. We also denote by $A\times\textnormal{Id}_{\measuredunitint}$ the graphing of dialect $\measured{D}^{A}\times\measuredunitint$ whose edges are realised as $\phi_{e}\times\textnormal{Id}_{\measuredunitint}$.
\end{notation}

\begin{definition}[Variants]\label{variant}
Let $F$ and $G$ be graphings. If there exists a Borel automorphism $\phi:\measuredunitint\rightarrow\measuredunitint$ such that $F=\textnormal{Id}_{\measured{X}}\times\phi(G)$, we say that $F$ and $G$ are variants.
\end{definition}

Morally, graphings are sort of graphs which offer richer combinatorics since two vertices might have a non-trivial intersection without being equal. In particular, when considering paths, one should be careful about the sources and targets: a path in a graphing $G$ is a sequence of edges $\pi=e_{1}, e_{2},\dots, e_{k}$ in $E^{G}$ such that not only $s^{G}(e_{i+1})\cap t^{G}(e_{i})$ is non-negligible for every $i$, but also verifying that every sequence 
$$(\phi^{G}_{e_{i}}\circ \phi^{G}_{e_{i-1}}\circ \dots \phi^{G}_{e_{1}})(s^{G}(e_{1}))$$
 is of strictly positive measure. This path is then naturally realised as the composite $\phi_{\pi}^{G}=\phi_{e_{k}}^{G}\circ\dots\circ\phi_{e_{1}}^{G}$, and is considered with its \emph{maximal domain} $s^{G}(\pi)$, i.e. the set of all $x$ such that for all $i$, $\phi_{e_{i}}^{G}\circ\dots\circ\phi_{e_{1}}^{G}(x)\in s^{G}(e_{i+1})$, and its codomain $t^{G}(\pi)=\phi_{\pi}^{G}(s^{G}(\pi))$. The weight of $\pi$ is obviously defined as $\omega^{G}(\pi)=\omega^{G}(e_{k})\omega^{G}(e_{k-1})\dots\omega^{G}(e_{1})$ using the composition law of $\Omega$.

We can then define \emph{alternating path} between thick graphings as in the case of graphs, and introduce the operation of \emph{execution} between thick graphings, the semantic counterpart to the cut-elimination procedure. We write $\mathrm{AltPath}(F,G)$ the set of all alternating paths between two graphings $F,G$. %We here define execution only in the specific case in which the support of one graphing is included in the support of the other: this represents an \emph{application}, i.e. a modus ponens; the general case of the cut can be consulted in our earlier paper \cite{seiller-goig} or deduced from this specific case\footnote{Noticing that $(A\multimap B)\otimes (B\multimap C)\multimap (B\multimap B)\multimap (A\multimap C)$ (using the semi-distributivity law twice), one can rewrite a cut between $f\in A\multimap B$ and $g\in B\multimap C$ as an application of $f\otimes g\in (B\multimap B)\multimap (A\multimap C)$ with the axiom $a\in B\multimap B$.}.

Before defining execution, we need to introduce an additional construction on paths that will allow us to restrict them to a subset of their domain, i.e. perform the \enquote{hiding part} of game semantics' composition. Given a path $\pi$ in a graphing $G$ and a measurable subset $C$ (thought of as the cut), we define ${}^{C}[\pi]_{o}^{o}$ as the path with same realiser and weight as $\pi$, and whose source has been restricted to the measurable set $s^{G}(\pi)\cap \bar{C}\cap (\phi^{G}_{\pi})^{-1}(\bar{C})$, where $\bar{C}$ is the complement set of $C$. Intuitively, we restrict $\pi$ to the maximal subset of its domain that lies outside of $C$ and whose image through the realiser $\phi^{G}_{\pi}$ lies outside of $C$. 

\begin{definition}[Execution]
Let $F$ and $G$ be graphings with $S^{F}=V\disjun C$ and $S^{G}=C\disjun W$. Their \emph{execution} $F\plug G$ is the graphing of support $V\disjun W$ defined as the set of all restrictions ${}^{C}[\pi]_{o}^{o}$ for alternating paths $\pi\in\mathrm{AltPath}(F,G)$.
\end{definition}

\begin{example}
We consider the two one-edge graphings (without dialects or weights to be concise) $G$ and $H$ illustrated on the left-hand side of Figure \ref{examplepaths}. The edge of $G$ has source the segment $[0,2]$, target the segment $[4,6]$ and is realised by the map $x\mapsto 6-x$. The edge of $H$ has source the segment $[5,6]$, target the segment $[8,9]$ and is realised by the map $x\mapsto x+3$. The \emph{cut} is represented by the segment $[5,6]$. The execution of $G$ and $H$, illustrated on the right-hand side of Figure \ref{examplepaths}, is composed of two paths: the restriction of the edge of $G$ to the segment $[1,2]$, and the composition of the two edges.
\end{example}

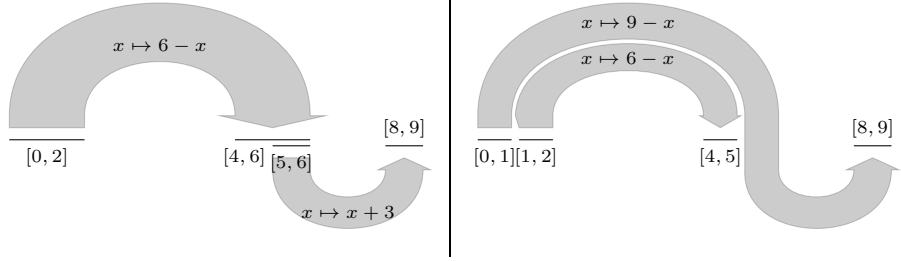
\begin{figure*}
\centering
\begin{tabular}{c|c}
\begin{tikzpicture}[x=1cm,y=0.8cm]
	\draw[-] (0,-0.4) -- (1,-0.4) node [midway, below] {\scriptsize{$[0,2]$}};
	\draw[-] (3,-0.4) -- (4,-0.4) node [very near start, below] {\scriptsize{$[4,6]$}};
	\draw[-,fill=black,opacity=0.2] (1,0) .. controls (1,1.2) and (3,1.2) .. node [midway,above,opacity=1] {\scriptsize{$x\mapsto 6-x$}} (3,0) -- (2.8,0) -- (3.5,-0.2) -- (4.2,0) -- (4,0) .. controls (4,2.5) and (0,2.5) .. (0,0) -- (0,-0.2) -- (1,-0.2) -- (1,0) {};
	
	\draw[-] (3.5,-0.5) -- (4,-0.5) node [midway, below] {\scriptsize{$[5,6]$}};
	\draw[-] (5,-0.5) -- (5.5,-0.5) node [midway, above] {\scriptsize{$[8,9]$}};
	\draw[-,fill=black,opacity=0.2] (4,-0.9) .. controls (4,-1.5) and (5,-1.5) .. (5,-0.9) -- (4.9,-0.9) -- (5.25,-0.7) -- (5.6,-0.9) -- (5.5,-0.9) .. controls (5.5,-2.2) and (3.5,-2.2) .. node [midway,above,opacity=1] {\scriptsize{$x\mapsto x+3$}} (3.5,-0.9) -- (3.5,-0.7) -- (4,-0.7) -- (4,-0.9) {};
\end{tikzpicture}
&
\begin{tikzpicture}[x=1cm,y=0.8cm]
	\draw[-] (0,-0.4) -- (0.45,-0.4) node [midway, below] {\scriptsize{$[0,1]$}};
	\draw[-] (0.55,-0.4) -- (1,-0.4) node [midway, below] {\scriptsize{$[1,2]$}};
	\draw[-] (3,-0.4) -- (3.45,-0.4) node [midway, below] {\scriptsize{$[4,5]$}};
	\draw[-,fill=black,opacity=0.2] (1,0) .. controls (1,1) and (3,1) .. (3,0) -- (2.95,0) -- (3.225,-0.2) -- (3.5,0) -- (3.45,0) .. controls (3.45,1.6) and (0.55,1.6) .. node [midway,below,opacity=1] {\scriptsize{$x\mapsto 6-x$}} (0.5,0) -- (0.55,-0.2) -- (1,-0.2) -- (1,0) {};
	\draw[-,fill=black,opacity=0.2] (0.45,0) .. controls (0.45,1.7) and (3.55,1.7) .. node [midway,above,opacity=1] {\scriptsize{$x\mapsto 9-x$}} (3.55,0) -- (3.55,-0.9) .. controls (3.5,-2.2) and (5.5,-2.2) .. (5.5,-0.9) -- (5.6,-0.9) -- (5.25,-0.7) -- (4.9,-0.9) -- (5,-0.9) .. controls (5,-1.5) and (4,-1.5) .. (4,-0.9) -- (4,0) .. controls (4,2.5) and (0,2.5) .. (0,0) -- (0,-0.2) -- (0.45,-0.2) -- (0.45,0) {};
	
	\draw[-] (5,-0.5) -- (5.5,-0.5) node [midway, above] {\scriptsize{$[8,9]$}};
\end{tikzpicture}
\end{tabular}
\caption{Example of an execution between two graphings.}\label{examplepaths}
\end{figure*}

Based on the notion of alternating cycle -- defined easily from the notion of alternating paths, one defines a measurement $\meas{\cdot,\cdot}$ of couples of graphings and taking values in $\realposext$. This measurement is parametrized\footnote{We won't dwell on this choice of parameter in this paper, in order to avoid unnecessary complications. Although we here mention it for the sake of exactness, it will not play any specific role here. A fine analysis of the models would imply a consideration of specific values of $m$, but none of the results obtained in this paper depend on the choice of $m$.} by the choice of a measurable map $m:\Omega\rightarrow\realposext$. It is a quite involved work to define and study, and the results of this paper are based only on the existence of such a measurement and not its definition, so we refer the interested reader to our previous paper \cite{seiller-goig}. In the specific case of graphs -- which are graphings over a discrete space -- this measurement simply equals the sum, over the set of alternating cycles $\pi$, of $m(\omega(\pi))$ where $m$ is any map $\Omega\rightarrow\realposext$. This notion of measurement is extended to couples $(a,A)$ where $a$ is a real number (potentially infinite) and $A$ a graphing; the consideration of this additional real number -- the wager -- finds its reasons in technical details that are explained in previous papers \cite{seiller-goim,seiller-goiadd}. The resulting couples, called projects, are used to interpret proofs.

\begin{definition}[Project]
A project is a pair $\de{a}=(a,A)$ with $a\in\realposext$ and $A$ is a formal weighted sum of graphings $A=\sum_{i\in I^{A}}\alpha^{A}_{i} A_{i}$. We write $\unit{A}$ the sum $\sum_{i\in I^{A}\alpha_{i}}$.
\end{definition}

From the measurement, one defines a notion of orthogonality that accounts for linear negation. This orthogonality relation is used to define \emph{conducts}, specific sets of projects which will interpret formulas.

\begin{definition}[Orthogonality]
Two projects $\de{a}=(a,A)$ and $\de{b}=(b,B)$ of equal supports are orthogonal, denoted $\de{a}\poll{}\de{b}$, when $a\unit{B}+b\unit{A}+\meas{A,B} \neq 0,\infty$.

Given a set $T$ of projects, its orthogonal $T^{\pol}$ is defined as $\{\de{a}~|~\forall\de{b}\in T, \de{a}\poll{}\de{b}\}$. We will denote $T^{\pol\pol}$ the set $(T^{\pol})^{\pol}$.
\end{definition}

\begin{definition}[Conduct]
A \emph{conduct} $\cond{A}$ of support $V$ is a set of projects of support $V$ which is bi-orthogonally closed: $\cond{A}=\cond{A}^{\pol\pol}$.
\end{definition}

Finally, one can define a category whose objects are conducts and morphisms are projects and which is shown to interpret multiplicative-additive linear logic. We do not detail this construction since it is quite involved. However, let us point out that the resulting model is completely non-degenerate (none of the connectives or constants are identified) and does not satisfy the mix and weakening rules \cite{seiller-goiadd}.

\begin{theorem}[Seiller \cite{seiller-goig}]\label{theoremMALL}
Let $\measured{X}$ be a measure space, $\microcosm{m}$ a microcosm, $\Omega$ a monoid of weights. For all measurable map $m:\Omega\rightarrow\realposext$, conducts and projects built from $\Omega$-weighted $\microcosm{m}$-graphings, with the orthogonality defined from the measurement defined from $m$, form a model of Multiplicative-Additive Linear Logic.
\end{theorem}

%Although we refer the reader to previous work \cite{seiller-goig} for details about the construction, we mention quickly how the connectconnectives 

\section{The model}

To describe the model, we will pick a measure space $\measured{X}$ together with a microcosm $\microcosm{ll}_{\rho}$ which are defined below. The construction we describe will not depend on the choices of $\Omega$ and $m:\Omega\rightarrow\realposN$, and therefore describes a family of quantitative models of second order linear logic.

Although the underlying space used here differs from our earlier work on exponentials \cite{seiller-goie}, both are equivalent up to a Borel automorphism. The presentation we chose to work with here has the advantage of showing more explicitly the dynamics at work, while gaining intuitions from standard work on exponentials. Indeed, we chose to work with the Hilbert cube $\measuredunitint^{\naturalN}$, underlying an intuitive correspondence\footnote{More precisely, the correspondence would be between boxes and copies of $\measuredunitint\times\measuredunitint$, cf. the definition of exponential connectives.} between \emph{boxes} used to treat exponentials in proof nets and the copies of $\measuredunitint$.

\begin{definition}[The space]
We define the measure space $\measured{X}=\realN\times\measuredunitint^{\naturalN}$, product of the real line with the Hilbert cube, endowed with its usual Borel algebra and Lebesgue measure.
\end{definition}

\begin{notation}
We will write elements of $\measured{X}$ as couples $(a,s)$, where $a\in\realN$ and $s$ is a sequence of elements in $\measuredunitint$. We will sometimes write sequences as $s\bullet s'$, i.e. as the concatenation of a finite sequence $s=(x_{1},\dots,x_{k})$ and a sequence $s'$; when $s$ contains only one element $x$ we will identify $x$ and $(x)$. When considering elements of the space $\measured{X}\times\measuredunitint$, we will use a natural extension of this notation, and write them $(a,s,e)$, with $(a,s)\in\measured{X}$ and $e\in\measuredunitint$.
\end{notation}

We now define the microcosm, denoted $\microcosm{ll}_{\rho}$, that will be used to interpret proofs. We could very well have worked with the biggest microcosm possible (the so-called \emph{macrocosm}) or any microcosm containing $\microcosm{ll}_{\rho}$. It is however more interesting to point out exactly the principles that are necessary to interpret second-order linear logic.

\begin{definition}[The microcosm]\label{microdef}
Let $\rho$ be a measure-preserving bijection $\measuredunitint^{2}\rightarrow \measuredunitint$. We define the microcosm $\microcosm{ll}_{\rho}$ as the monoid of measurable\footnote{We notice that those are all Borel automorphisms, thus in particular Borel-preserving and non-singular.} maps $\measured{X}\times\measuredunitint\rightarrow\measured{X}\times\measuredunitint$ generated by:
\begin{itemize}[nolistsep,noitemsep]
\item affine transformations on $\realN$: $A_{\lambda}^{\alpha}: (x,s,e)\mapsto (\alpha x+\lambda,s,e)$; 
\item (finite) permutations on $\measuredunitint^{\naturalN}$: $P_{\sigma}: (x,s,e)\mapsto (x,\sigma(s),e)$;
\item the map $D_{\rho}: (a,(x,y)\bullet s,e)\mapsto (a,\rho(x,y)\bullet s,e)$ and its inverse;
\item the \emph{exchange} $\exchange$: $(a,x\bullet s,e)\mapsto (a, e\bullet s, x)$
\end{itemize}
\end{definition}

%\begin{remark}
%The actual interpretation of proofs will also make use of a measure-preserving bijection $\theta:\measuredunitint^{3}\rightarrow\measuredunitint$. This measure-preserving bijection can be constructed from $\rho:\measuredunitint^{2}\rightarrow\measuredunitint$ by composition, e.g. $\theta=\rho\circ (\rho\times\textnormal{Id}_{\measuredunitint})$, and the microcosm $\microcosm{ll}_{\rho}$ defined above therefore contains such a map already. However, the interpretation of proofs could very well use a completely independent map $\theta$; this would give a more general construction of models of linear logic but we chose here the more compact construction where $\theta$ is defined as $\rho\circ (\rho\times\textnormal{Id}_{\measuredunitint})$.
%\end{remark}

Notice that the exchange $\exchange$ is an example of map that could not arise from a microcosm of maps from $\measured{X}$ to itself. This added principle is crucial for the definition of both dereliction and digging. Intuititvely, the microcosm of Definition \ref{microdef} \emph{without the exchange map} allows for Elementary Linear Logic\footnote{\label{footnote}To be more exact, the microcosm allowing for a model of ELL is the microcosm $\microcosm{ll}_\rho$ without the exchange but with the maps $D_{\sigma}$ which permute the family of intervals $\{[(i-1)/k,i/k]\}_{i=1}^{k}$ in the dialect along a permutation $\sigma$ of $\{1,\dots,k\}$. Without these maps, one cannot define contraction as one cannot represent \emph{slice-changing edges} \cite{seiller-goie}; it is not necessary to have all of them, though, as for instance all such $D_{\sigma}$ for permutations $\sigma$ over sets $\{1,\dots, 2^{p}\}$ are enough. Notice that these maps -- in the case $k=2^{p}$ -- are elements of $\microcosm{ll}_\rho$, defined as $D_{\sigma}=\exchange\circ\rho_{(p)}\circ P_{\sigma}\circ\rho^{-1}_{(p)}\circ\exchange$, where $\rho_{(p)}$ is recursively defined by:\newline
\begin{center}\vspace{-0.4cm}$\rho_{(0)}=\rho\hspace{2cm}\rho_{(p)}=\rho_{(p-1)}\circ(\prod_{i=1}^{2^{p-1}}\rho)$\end{center}\vspace{-0.4cm}}, in the same spirit as our previous work on exponentials \cite{seiller-goie}; the added principle -- the exchange -- adds both dereliction and digging simultaneously%\footnote{Fortunately! Obtaining a model with exactly one of those two principles is incoherent \cite{blindspot}.}
.

\begin{remark}
One actually considers thick and sliced graphings up to a larger equivalence than that of variants. Indeed, the sliced and thick graphing $\sum_{i=1}^{k} \frac{1}{k}A_{i}$ is considered equivalent to the \emph{universal}\footnote{This is the \emph{smallest} such graphing, i.e. if $H'$ also satisfies this property, then $H$ is included in $H'$.} graphing $H$ whose restriction to the part of the dialect $[(i-1)/k,i/k]$ is equal to $A_{i}$, modulo the affine transformation $[(i-1)/k,i/k]\rightarrow [0,1], x\mapsto (x\times k)-i+1$.
\end{remark}
%\note{Identification between $\sum_{i=1}^{k} A_{i}$ where $A_{i}$ is of dialect $[0,1]$ and the graphing $G$ of dialect $[0,1]$ such that the restriction of $G$ to the dialect $[(i-1)/k,i/k]$ is equal to $A_{i}$.}

By Theorem \ref{theoremMALL}, we know that for any choices of $\Omega$ and $m$, the induced model interprets \MALL. We will thus concentrate on exponential connectives here and refer the interested reader to earlier papers for the definition of \MALL connectives.

\section{The Exponentials}

We now define the perennisation, that is the operation turning a project $\de{a}$ into a project  $\oc \de{a}$ that can be duplicated. Indeed, in the interaction graphs models a project $\de{ctr}$ interpreting the contraction of arguments can be defined but it actually implements contraction, i.e.\ satifsfies $\de{ctr}\plug \de{a}\equiv\de{a}\otimes \de{a}$ for a natural notion of equivalence $\equiv$, \emph{only if} the graphing $G$ is dialect-free \cite{seiller-phd,seiller-goie}. Thus the need for a \emph{perennisation} operation that turns a project $\de{a}$  into a dialect-free project $\oc \de{a}$; this operation will in turn, when lifted to conducts, define the exponential connective.

In order to preserve all information contained in the dialect, the operation on graphings that underlies the perennisation will encode the information contained in the dialect in the support of the graphing. It is to be noted that the perennisation operation is not defined on all projects, but on the subset of so-called \emph{balanced} projects. These are in particular projects whose dialect is equal to $\measuredunitint$ or, by extension, whose graphings are \enquote{balanced} sums: $\sum_{i=1}^{k} \alpha_{i}A_{i}$ such that $\alpha_{i}=1/k$ for all $i$. This does not hinder the interpretation of proofs since projects arising from such an interpretation will all satisfy these conditions. Once this restriction is considered, the operation itself is easy to define: from a balanced project $(0,A)$ we construct $(0,\oc A)$ where $\oc A$ is the graphing obtained by \enquote{pushing} the dialect of $A$ into the first slot of the sequence $\measuredunitint^{\naturalN}$. For technical reasons explained later, we also need to create a fresh new copy of $\measuredunitint$ that will be used for implementing the promotion rule.

\begin{definition}
A project $\de{a}=(a,A)$ is \emph{balanced} if $a=0$ and the dialect of $A$ is $\measuredunitint$. If $E$ is a set of projects, we write $\balanced{E}$ the subset of balanced projects in $E$.
\end{definition}

%Notice this differs from the former definition of balanced projects that appears in our work on Elementary Linear Logic \cite{seiller-goie}. It is needed to consider that balanced project have a continuous dialect $\measuredunitint$ since some projects used to interpret proofs (e.g. dereliction and digging) will have continuous dialects. Some projects, however, may be built from a graphing $G$ having a discrete dialect $D$, e.g. the interpretation of exponential-free proofs. As already noted above, those have variants $G'$ of continuous dialect $\measuredunitint$, hence this notion of balanced project thus extend the notion introduced in previous work.

In order to define exponentials, we will need the following map:
$$B:\left\{\begin{array}{rcl}\measured{X}\times\measuredunitint&\rightarrow&\measured{X}\\
					(a,s,d)&\mapsto&(a,d\bullet s)\end{array}\right.$$
This map $B$ will be our way of encoding the dialect of $A$ in the support of the graphing $\oc A$. This way, the resulting graphing $\oc A$ will contain the exact same information as $A$, but will be dialect-free. Though it might seem a transparent and useless operation, the fact that the dialect is now part of the support makes the graphings $\oc A$ and $A$ behave quite differently when put into interaction with other projects. Intuitively, while the dialect is something private -- e.g. control states -- the support is not, and some projects might interact with $\oc A$ non-uniformly w.r.t. the former dialect of $A$. 

\begin{example}\label{exampleoc}
We consider two graphings, say $G$ and $H$, both though of\footnote{Recall that we are actually working with \enquote{variants} $G^{c}$ and $H^{c}$ whose dialect are $\measuredunitint$.} as graphings with dialects $\{1,2\}$. Now, suppose that $G$ and $H$ are of type $\cond{A}$ and $\cond{A\multimap B}$ respectively. Their execution $F\plug G$ is then of type $\cond{B}$, and its dialect should be thought of as $\{1,2\}\times\{1,2\}$. We can also consider $\oc G$ and $\oc H$, which are of respective types $\cond{\oc A}$ and $\cond{\oc(A\multimap B)}$, and their execution $\oc G\plug \oc H$. Let us explain why the latter cannot be of type $\cond{\oc B}$. Figure \ref{figureexampleoc} illustrates this situation with examples of graphings $G$, $\oc G$, $H$, $\oc H$, as well as lists of the edges (alternating paths) of $G\plug H$ and $\oc G\plug \oc H$.

The execution of $\oc G$ and $\oc H$ actually produces the graphing defined as follows: compute the execution of $G$ and $H$ \emph{as if they did not have any dialect}, and then take the perennisation of the result. In other words, the only alternating paths computed between $G$ and $H$ are those where the states of $G$ and $H$ are equal: this creates new paths (pictured in red path in Figure \ref{figureexampleoc}), this deletes paths (the blue paths), and leaves some of them \enquote{unchanged}. As a consequence, we cannot prove that $\oc G\plug \oc H$ is of type $\cond{\oc B}$ since the only graphing we know for sure to belong to this type is $\oc(G\plug H)$.
\end{example}

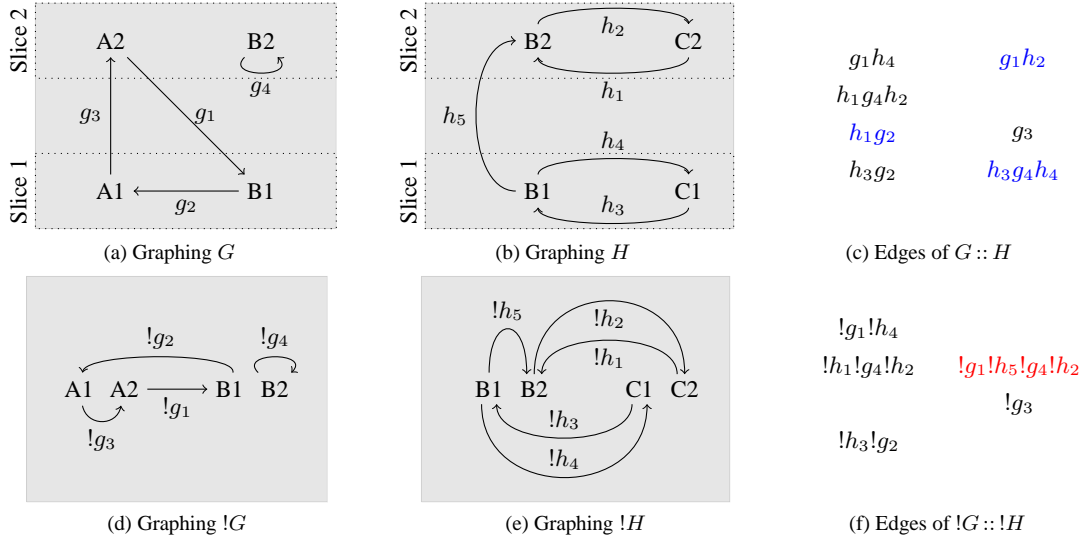
\begin{figure*}
\centering
\subfloat[Graphing $G$]{
\begin{tikzpicture}
	\node (A1) at (0,0) {A1};
	\node (A2) at (0,2) {A2};
	\node (B1) at (2,0) {B1};
	\node (B2) at (2,2) {B2};

	\draw[dotted] (-1,0.5) -- (3,0.5) -- (3,-0.5) -- (-1,-0.5) -- (-1,0.5) node [midway,above,sloped] {Slice 1};
	\draw[dotted] (-1,2.5) -- (3,2.5) -- (3,1.5) -- (-1,1.5) -- (-1,2.5) node [midway,above,sloped] {Slice 2};
	\draw[-,fill=black,opacity=0.1] (-1,-0.5) -- (-1,2.5) -- (3,2.5) -- (3,-0.5) -- (-1,-0.5) {};
	\draw[->] (A2) -- (B1) node [midway,right] {$g_1$};
	\draw[->] (B1) -- (A1) node [midway,below] {$g_2$};
	\draw[->] (A1) -- (A2) node [midway,left] {$g_3$};
	\draw[->] (B2) .. controls (1.5,1.5) and (2.5,1.5) .. (B2) node [midway,below] {$g_4$};
\end{tikzpicture}
}\hspace{0.5cm}
\subfloat[Graphing $H$]{
\begin{tikzpicture}
	\node (B1) at (6,0) {B1};
	\node (B2) at (6,2) {B2};
	\node (C1) at (8,0) {C1};
	\node (C2) at (8,2) {C2};
	
	\draw[dotted] (4.5,0.5) -- (8.6,0.5) -- (8.6,-0.5) -- (4.5,-0.5)-- (4.5,0.5)  node [midway,above,sloped] {Slice 1};
	\draw[dotted] (4.5,2.5) -- (8.6,2.5) -- (8.6,1.5) -- (4.5,1.5)-- (4.5,2.5) node [midway,above,sloped] {Slice 2};
	\draw[-,fill=black,opacity=0.1] (4.5,-0.5) -- (4.5,2.5) -- (8.6,2.5) -- (8.6,-0.5) -- (4.5,-0.5) {};
	\draw[->] (C2) .. controls (8,1.5) and (6,1.5) .. (B2) node [midway,below] {$h_1$};
	\draw[<-] (C2) .. controls (8,2.5) and (6,2.5) .. (B2) node [midway,below] {$h_2$};
	\draw[->] (C1) .. controls (8,-0.5) and (6,-0.5) .. (B1) node [midway,above] {$h_3$};
	\draw[<-] (C1) .. controls (8,0.5) and (6,0.5) .. (B1) node [midway,above] {$h_4$};
	\draw[->] (B1) .. controls (5,0) and (5,2) .. (B2) node [midway,left] {$h_5$};
\end{tikzpicture}
}\hspace{0.5cm}
\subfloat[Edges of $G\plug H$]{
\begin{tikzpicture}
	\node (1) at (2,1.5) {\textcolor{blue}{$g_1 h_2$}};
	\node (2) at (0,1.5) {$g_1 h_4$};
	\node (3) at (2,1) {};
	\node (4) at (0,1) {$h_1 g_4 h_2$};
	\node (5) at (2,0.5) {$g_3$};
	\node (6) at (0,0.5) {\textcolor{blue}{$h_1 g_2$}};
	\node (7) at (2,0) {\textcolor{blue}{$h_3 g_4 h_4$}};
	\node (8) at (0,0) {$h_3 g_2$};
	\node (SW) at (-1,-0.7) {};
	\node (NE) at (2.5,2.2) {};
\end{tikzpicture}
}
\\\hspace{0.5cm}
\subfloat[Graphing $\oc G$]{
\begin{tikzpicture}	
	\node (A1B) at (-0.3,-3) {A1};
	\node (A2B) at (0.3,-3) {A2};
	\node (B1B) at (1.7,-3) {B1};
	\node (B2B) at (2.3,-3) {B2};
	
	\draw[-,fill=black,opacity=0.1] (-1,-4.5) -- (-1,-1.5) -- (3,-1.5) -- (3,-4.5) -- (-1,-4.5) {};
	\draw[->] (A2B) -- (B1B) node [midway,below] {$\oc g_1$};
	\draw[->] (B1B) .. controls (1.8,-2.5) and (-0.2,-2.5) .. (A1B) node [midway,above] {$\oc g_2$};
	\draw[->] (A1B) .. controls (-0.2,-3.5) and (0.2,-3.5) .. (A2B) node [midway,below] {$\oc g_3$};
	\draw[->] (B2B) .. controls (1.8,-2.5) and (2.8,-2.5) .. (B2B) node [midway,above] {$\oc g_4$};
\end{tikzpicture}
}\hspace{1cm}
\subfloat[Graphing $\oc H$]{
\begin{tikzpicture}
	\node (B1B) at (5.7,-3) {B1};
	\node (B2B) at (6.3,-3) {B2};
	\node (C1B) at (7.7,-3) {C1};
	\node (C2B) at (8.3,-3) {C2};
	
	\draw[-,fill=black,opacity=0.1] (4.8,-4.5) -- (4.8,-1.5) -- (8.9,-1.5) -- (8.9,-4.5) -- (4.8,-4.5) {};
	\draw[->] (8.2,-2.8) .. controls (8.2,-2.2) and (6.4,-2.2) .. (6.4,-2.8) node [midway,below] {$\oc h_1$};
	\draw[<-] (8.3,-2.8) .. controls (8.3,-1.5) and (6.3,-1.5) .. (6.3,-2.8) node [midway,below] {$\oc h_2$};
	\draw[->] (7.6,-3.2) .. controls (7.6,-3.8) and (5.8,-3.8) .. (5.8,-3.2) node [midway,above] {$\oc h_3$};
	\draw[<-] (7.8,-3.2) .. controls (7.8,-4.5) and (5.6,-4.5) .. (5.6,-3.2) node [midway,above] {$\oc h_4$};
	\draw[->] (B1B) .. controls (5.7,-2) and (6.2,-2) .. (6.2,-2.8) node [midway,above] {$\oc h_5$};
\end{tikzpicture}
}\hspace{0.5cm}
\subfloat[Edges of $\oc G \plug \oc H$]{
\begin{tikzpicture}
	\node (1) at (0,1) {$\oc h_1 \oc g_4 \oc h_2$};
	\node (2) at (0,1.5) {$\oc g_1 \oc h_4$};
	\node (3) at (2,1) {\textcolor{red}{$\oc g_1 \oc h_5 \oc g_4 \oc h_2$}};
	\node (4) at (2,0.5) {$\oc g_3$};
	\node (5) at (0,0) {$\oc h_3 \oc g_2$};
	\node (SW) at (-1,-0.7) {};
	\node (NE) at (2.5,2.2) {};
\end{tikzpicture}
}
\caption{Illustration of Example \ref{exampleoc} with graphing seens as graphs.}
\end{figure*}

%The perennisation works as follows: we take a graphing $A$ with dialect $\measuredunitint$. Its edges are therefore realised as measurable maps $\measured{X}\times\measuredunitint\rightarrow \measured{X}\times\measuredunitint$. We then consider the graphing $A^{\uparrow}$ with dialect $\measuredunitint\times\measuredunitint$ whose edges $e\in E^{A^{\uparrow}}=E^{A}$ are realised as $\phi^{A}_{e}\times\measuredunitint$, i.e. as measurable maps $\measured{X}\times\measuredunitint^{2}\rightarrow\measured{X}\times\measuredunitint^{2}$. Finally, we define the perennialized graphing $\oc A$ as the graphing obtained from \enquote{burying} the unit square $\measuredunitint$ into $\measured{X}$ by using $B$, i.e. $\oc A$ is defined as the graphing whose edges $e\in E^{\oc A}=E^{A}$ are realised as $\phi^{\oc A}=B^{-1}\circ B^{-1}\circ \phi^{A^{\uparrow}}\circ B\circ B$.

\begin{definition}[Perennisation]
Let $\de{a}=(0,A)$ be a balanced project. We define its \emph{perennisation} $\de{\oc a}=(0,\oc A)$ by considering the dialect-free graphing $\oc A= B^{2}(A\times \textnormal{Id}_{\measuredunitint})$.
\end{definition}

\begin{definition}[Exponentials]
Let $\cond{A}$ be a conduct. We define the perennial conduct $\cond{\oc A}$ as the bi-orthogonal closure $\cond{\oc A}=\cond{(\sharp A)^{\pol\pol}}$ where $\cond{\sharp A}$ is the set
$$\cond{\sharp A}=\{\de{\oc a}~|~\de{a}\in\balanced{\cond{A}}\}$$
\end{definition}

%\textcolor{red}{Other design choice.\\
%We can chose to work with graphings whose vertices are realised as measurable subsets of $\realN\times\measuredunitint^{\naturalN}$, but we could as well use measurable subsets of $\realN\times\measuredunitint^{k}$, where $k$ can be any integer. This would make the definition of perennisation particularly compact (it would simply be mapping a graphing of dialect $\measuredunitint$ to the same graphing with dialect $\{\star\}$. Functorial promotion is not completely obvious since it needs to go through some isomorphisms between the unit interval and the unit square; this can be dealt with by considering $\measuredunitint\times\measuredunitint$ instead of $\measuredunitint$ in the base measure space. This would not change anything though. We chose not to follow this design choice here as it yields a sort of \emph{non-uniform} model, i.e. the stratification is somehow encoded in the base space in which some specific locations are designed to represent specific proofs/formulas.
%}

\section{A Model of Full Linear Logic}

We already know from previous work that the model just described is a model of multiplicative-additive linear logic with second-order quantification \cite{seiller-goig}. To ensure that we have a sound interpretation of exponential connectives, we will show that the following principles can be implemented:
\begin{itemize}[nolistsep,noitemsep]
\item functorial promotion $(\oc A\otimes \oc(A\multimap B))\multimap \oc B$;
\item dereliction $\oc A\multimap A$;
\item digging $\oc A\multimap \oc \oc A$.
\end{itemize}
The principle of \emph{contraction} $\oc A\multimap \oc A\otimes \oc A$ does not appear in this list as it holds for every possible definition of perennisation\footnote{As explained in Footnote \ref{footnote}, the microcosm already contains all the needed maps to define contraction.}. Let us notice moreover that the principle of functorial promotion was already shown to hold in our earlier work on exponentials \cite{seiller-goie}. We will however use here a less involved method for defining exponentials and implementing functorial promotion. The principles at work are more or less the same as in our earlier work, but this new implementation -- inspired from recent work on complexity \cite{seiller-towards} -- offers a clearer picture.

The change of perspective illustrated in Example \ref{exampleoc} is at the heart of the question of implementing functorial promotion. We want to \enquote{simulate} the disjointness of dialects. This is done in two steps: first make the encodings (in $\oc G$ and $\oc H$) of the dialects of $G$ and $H$ disjoint, by linking $\oc G$ and $\oc H$ through the permutation exchanging the two first copies of $[0,1]$. This corresponds to encoding the dialect of one of the two graphings on the second copy of $[0,1]$ instead of the first. Then we compute the result of this execution, obtaining a graphing which is almost $\oc (G\plug H)$ except for the fact that its dialect is encoded on the two first copies of $[0,1]$ and not only on the first. We then use a specific graphing that will use the map $D_{\rho}$ to encode this dialect on the first copy only.

\begin{theorem}\label{functorial}
Functorial Promotion holds.
\end{theorem}

\begin{proof}[Proof (Sketch)]
The proof is much simpler in this setting than in our previous work on exponentials \cite{seiller-goie}. The principle is however quite the same: we use a first map to ensure the disjointness of the two \enquote{public dialects}, and then we use a second map that will merge both copies. I.e. we define the maps:
$$
\begin{array}{rcrcl}
\textnormal{twist}&:& (\lambda,(x,\rho(y,z))\bullet s)&\mapsto& (\lambda,(y,\rho(x,z))\bullet s)\\
\textnormal{merge}&:& (\lambda, (x, \rho(y,z))\bullet s)&\mapsto& (\lambda,(\rho(x,y),z)\bullet s)
\end{array}
$$
To prove the result, we exhibit a project $\de{prom}$ and show that $\de{prom}\in\cond{\oc A\otimes\oc (A\multimap B)\multimap\oc B}$. For this, we show that for all $\de{\oc a}=(0,\oc A)\in\cond{\oc A}$ and $\de{\oc f}=(0,\oc F)\in\cond{\oc (A\multimap B)}$, we have
$$\de{prom}\plug\de{\oc a}\plug\de{\oc f}=(0,\textnormal{merge}(\oc A\plug \textnormal{twist}(\oc F))$$
Finally, one easily checks that $\textnormal{merge}(\oc A\plug \textnormal{twist}(\oc F))$ is equal to $\oc D$ where $D$ is a variant of $F\plug A$.
\end{proof}

Both digging and dereliction will work based on the simple idea that a continuous dialect $\measuredunitint$ can be exchanged with a copy of $\measuredunitint$ appearing in the Hilbert cube. This is exactly the computational principle encapsulated in the exchange map $\exchange$. This implies that the \emph{potential infinite} of dialects -- i.e. the fact that a dialect can be any finite set, without bounds on its cardinality -- can be managed within the projects themselves, something that could not be done in earlier constructions.
%For instance, if a finite dialect $\mathbf{k}:=\{1,\dots,k\}$ is encoded into $\measuredunitint$, as explained above.
%\note{These two principles are almost inverse one of the other, which is not surprising when considering that they form the coproduct and counit of a comonad.}

\begin{theorem}\label{digging}
Digging holds.
\end{theorem}

\begin{proof}[Proof (Sketch)]
As for the proof of Theorem \ref{functorial}, we exhibit an element $\dig{A}\in\cond{\oc A\multimap \oc \oc A}$. We show that, for all $\de{\oc a}=(0,\oc A)$ in $\cond{\oc A}$, one can compute $\dig{A}\plug\de{\oc a}=(0,\textnormal{push}(\oc A))$, where:
$$\textnormal{push}: (\lambda, (x,\rho(y,\rho(z,w)))\bullet s,e) \mapsto (\lambda; (e,z,x,y)\bullet s, w)$$
It is clear that $\textnormal{push}(\oc A)$ is equal to $\oc\oc A$, since the dialect of $\oc A$ (although $\oc A$ is dialect-free, not all elements of $\cond{\oc A}$ are, and therefore this is important) is encoded in the first copy of $\measuredunitint$, while the second copy remains unused (this is because the second copy of $\measuredunitint$ in $\oc A$ is unused\footnote{As this is not the case for elements of $\cond{A}$, this explains why digging is not a co-dereliction.}).
\end{proof}

The dereliction consists in \enquote{reconstructing} a dialect from a banged project. This can be performed using the same kind of tricks, i.e. using a continuous dialect.

\begin{theorem}\label{dereliction}
Dereliction holds.
\end{theorem}

\begin{proof}[Proof (Sketch)]
Again, we exhibit an element $\der{A}\in\cond{\oc A\multimap A}$. For this, we show that for all $\de{\oc a}\in\cond{\oc A}$, one can compute $\der{A}\plug \de{\oc a}=(0,\textnormal{raise}(\oc A))$ where
$$\textnormal{raise}:(\lambda,(x,y)\bullet s, e)\mapsto (\lambda,s,\rho(x,\rho(y,e)))$$
Again, one easily checks that $\textnormal{raise}(\oc A)$ is equal to $A$.
\end{proof}

Notice that the maps used to interpret dereliction and digging are not the only ones that satisfy the right properties, e.g. if one replaces $\textnormal{raise}$ by the map $\textnormal{raise}^{(2)}(\lambda,(x,y)\bullet s,e)\mapsto (\lambda,s,\rho(\rho(x,y),e))$, we still have $\textnormal{raise}^{(2)}$. The exact expressions are however important when to ensure that the execution soundly represents cut-elimination.

%Before going into some details about the interpretation of logic and how the model behaves w.r.t. cut-elimination, we ought to discuss what sort of models we just defined. Let us recall the structure of the derived categorical models for \MALL. There are in fact two models one can define from our constructions \cite{seiller-goiadd}: the first is the category \catmll of conducts with projects as morphisms, the second is the category \concat of conducts with equivalence classes of projects as morphisms (for the observational equivalence). Both categories are $*$-autonomous categories, and one can find a full subcategory of \concat with products, noted \behcat: here the quotient is essential because the equivalence in \autoref{thmcutadd} cannot be replaced by an equality. An equivalent way of saying this is that \behcat is the quotient of a full subcategory \catmall of \catmll which do not have products and coproducts. What we have shown in this section are results concerned with the non-quotiented category \catmll: we have shown that $\oc$ defines a functor on \catmll (\autoref{functorial}), and that this functor has the structure of a (monoidal) comonad. Indeed, the interpretation of digging and dereliction exposed in the proofs of \autoref{digging} and \autoref{dereliction} give rise to natural transformations. In the next section, we will first show a soundness result, and then study hat becomes of the constructions in the quotiented category \concat.

\section{Interpretation of proofs}\label{section:interpretation}

We first recall the notion of winning projects \cite{seiller-goie}. Winning projects are the equivalent of game semantics' winning strategies or classical realisability's proof-like terms. In particular, all interpretations of proofs will be winning projects.

\begin{definition}
A project $\de{a}=(a,A)$ is \emph{winning} if it is balanced and if $A$ is a disjoint union of transpositions, i.e. each edge $e$ in $A$ has a reverse edge $e^{\ast}$ with $\phi_{e^{\ast}}^{A}=(\phi_{e}^{A})^{-1}$ and the sources of edges are pairwise disjoint.
\end{definition}

We now recall the basics of the proof system for which we define the interpretation of proofs. We are working with three different kinds of formulas, \emph{positive}, \emph{negative} and \emph{neutral}. The technical reasons behind this are explained in our work on \ELL \cite{seiller-goie}. Intuitively, neutral formulas correspond to the fragment of linear logic which does not allow for structural rules, negative formulas are those created from a perennisation while positive formulas are duals of negative formulas. They are defined inductively through the grammar shown on Figure \ref{formulas} (neutral formulas are denoted by $B$ which stands for \emph{behavior} \cite{seiller-goie}).
\begin{figure*}
\centering
\begin{eqnarray*}
B&\!\!\!:=\!\!\!& X~|~X^{\pol}~|~\cond{0}~|~\cond{T}~|~B\otimes B~|~B\parr B~|~B\oplus B~|~B\with B~|~\forall X~B~|~\exists X~B~|~B\otimes N~|~B\parr P\\
N&\!\!\!:=\!\!\!& \cond{1}~|~\oc B~|~\oc N~|~N\otimes N~|~ N\with N~|~N\oplus N~|~N\parr P~|~\forall X~N~|~\exists X~N\\
P&\!\!\!:=\!\!\!&  \cond{\bot}~|~\wn B~|~\wn P~|~P\parr P~|~ P\with P~|~P\oplus P~|~N\otimes P~|~\forall X~P~|~\exists X~P
\end{eqnarray*}
\caption{Grammar for the formulas of $\textnormal{LL}_{\textnormal{pol}}$ (the symbol $X$ denotes variables)\label{formulas}}
\end{figure*}

\begin{definition}
A sequent $\Delta\pdash \Gamma;\Theta$ is such that $\Delta,\Theta$ contain only negative formulas, $\Theta$ containing at most one formula and $\Gamma$ containing only neutrals.
\end{definition}

\begin{definition}[The System $\textnormal{LL}_{\textnormal{pol}}$]
A proof in the system $\textnormal{LL}_{\textnormal{pol}}$ is a derivation tree constructed from the derivation rules of $\textnormal{ELL}_{\textnormal{pol}}$ \cite{seiller-goie}, which are nothing more than polarised variants of elementary linear logic sequent calculus rules -- presented with functorial promotion, extended with the rules in Figure \ref{llpol}.
\end{definition}

\begin{figure*}
\centering
\begin{tabular}{c}
\begin{tabular}{ccc}
\begin{minipage}{3.5cm}
\begin{prooftree}
\AxiomC{$\Delta, N\pdash \Gamma;\Theta$}
\RightLabel{\scriptsize{der$^{pol}$}}
\UnaryInfC{$\Delta, \oc N\pdash\Gamma;\Theta$}
\end{prooftree}
\end{minipage}
&
\begin{minipage}{3.5cm}
\begin{prooftree}
\AxiomC{$\Delta\pdash \Gamma,B;\Theta$}
\RightLabel{\scriptsize{der}}
\UnaryInfC{$\Delta, \oc B^{\pol}\pdash\Gamma;\Theta$}
\end{prooftree}
\end{minipage}
&
\begin{minipage}{3.cm}
\begin{prooftree}
\AxiomC{$\Delta, \oc \oc N\pdash \Gamma;\Theta$}
\RightLabel{\scriptsize{dig}}
\UnaryInfC{$\Delta, \oc N\pdash\Gamma;\Theta$}
\end{prooftree}
\end{minipage}
\end{tabular}
\end{tabular}
\caption{Additional Rules for Dereliction and Digging}\label{llpol}
\end{figure*}

One can then extend the inductive interpretation of proofs defined for $\textnormal{ELL}_{\textnormal{pol}}$ in earlier work \cite{seiller-goie} by interpreting the additional rules as follows: the interpretation $\Int{\pi}{}$ of a proof $\pi$ obtained from a proof $\pi'$ by using a dereliction rule (resp. a digging rule) on $\cond{\oc A}$ is defined as the execution of $\Int{\pi'}{}$ with the project $\der{A}$ (resp. $\dig{A}$).

\begin{theorem}
For every proof $\pi$ of a sequent $\Delta\pdash \Gamma;\Theta$ in $\textnormal{LL}_{\textnormal{pol}}$, the interpretation $\Int{\pi}{}$ is a winning project in $\Int{\Delta\pdash \Gamma;\Theta}{}$
\end{theorem}

The proof of this result is uninteresting in itself and follows exactly the proof of the same result for the restricted system $\textnormal{ELL}_{\textnormal{pol}}$ \cite{seiller-phd,seiller-goie}. The additional cases of dereliction and digging rules are completely transparent since the projects exhibited in the proofs of Theorems \ref{digging} and \ref{dereliction} are clearly winning projects.

Notice that it is an open question whether the exponential isomorphism between $\oc (A\with B)$ and $\oc A \otimes \oc B$ did hold in the ELL model\footnote{This is discussed in our earlier paper \cite{seiller-goie}, but can be understood as follows: in the non-affine sequent calculus for ELL (presented with the functorial promotion rule) one cannot prove the implication $(\oc A\otimes \oc B)\multimap\oc (A\with B)$.}. In the model we just described, however, this isomorphism holds; one just has to write down the usual derivation which can be interpreted soundly in the model.
\begin{prooftree}
\AxiomC{$\pdash A,A^{\pol{}};$}
\RightLabel{\scriptsize{der}}
\UnaryInfC{$\oc A\pdash A;$}
\RightLabel{\scriptsize{weak}}
\UnaryInfC{$\oc A,\oc B\pdash A;$}
\AxiomC{$\pdash B,B^{\pol{}};$}
\RightLabel{\scriptsize{der}}
\UnaryInfC{$\oc B\pdash B;$}
\RightLabel{\scriptsize{weak}}
\UnaryInfC{$\oc A,\oc B\pdash B;$}
\RightLabel{\scriptsize{$\with$}}
\BinaryInfC{$\oc A,\oc B\pdash A\with B;$}
\RightLabel{\scriptsize{$\oc$}}
\UnaryInfC{$\oc\oc A,\oc\oc B\pdash ;\oc (A\with B)$}
\RightLabel{\scriptsize{dig}}
\UnaryInfC{$\oc A,\oc\oc B\pdash ;\oc (A\with B)$}
\RightLabel{\scriptsize{dig}}
\UnaryInfC{$\oc A,\oc B\pdash ;\oc (A\with B)$}
\doubleLine
\UnaryInfC{$\pdash;(\oc A\otimes\oc B)\multimap\oc (A\with B)$}
\end{prooftree}

This implies not only that that two conducts $\oc A\otimes\oc B$ and $\oc (A\with B)$ are isomorphic, but that they are equal. Indeed, the inclusion $\oc (A\with B)\subseteq\oc A\otimes\oc B$ can be proved as in our earlier paper \cite{seiller-goie}. Moreover, the interpretation $\Int{\pi}{}$ of the above derivation can be shown to satisfy $\Int{\pi}{}\plug(\de{\oc a}\otimes\de{\oc b})=\de{\oc (a\with b)}$, where $\de{a\with b}$ is the usual construction of the $\with$ rule between $\de{a}$ and $\de{b}$, yielding the converse inclusion.

\begin{theorem}\label{expiso}
For any conducts $\cond{A}$ and $\cond{B}$, $\cond{\oc (A\with B)}=\cond{\oc A\otimes\oc B}$.
\end{theorem}

Before discussing the interpretation of cut-elimination in the models, we should add a few explanations about the soundness results we just obtained. Indeed, it should be stressed that the interpretation of proofs we just defined are non-trivial, i.e. that they do not identify (too much) distinct proofs. In the case of our models this is a consequence of the fact that two proofs have the same interpretation if and only if they have the same proof net \cite{linearlogic,ProofNetHilbert}. In other words, sequent calculus proofs are quotiented w.r.t. some (computationally inessential) commutations of rules.

\paragraph{What about cut-elimination?}
It is known that GoI does not represent cut-elimination exactly, i.e. it is not always the case that if $\pi'$ is the normal form of $\pi$, then $\Int{\pi'}{}=\textnormal{Ex}(\Int{\pi}{})$. It was shown that cut-elimination for \MLL is soundly represented by execution \cite{seiller-goim}, but there is a mismatch even in the exponential-free fragment because of additive cuts. This issue is discussed in details in previous work \cite{seiller-goiadd}, where we solve this problem by considering a notion of \emph{observational equivalence} $\cong$. Indeed, we showed that, even if $\Int{\pi'}{}\neq\textnormal{Ex}(\Int{\pi}{})$ in presence of an additive cut, $\Int{\pi'}{}$ and $\textnormal{Ex}(\Int{\pi}{})$ are observationally equivalent.

\begin{theorem}[Seiller \cite{seiller-goiadd}]\label{thmcutadd}
If $\pi'$ is obtained from $\pi$ by applying a step of cut-elimination ($\with$/$\oplus$) then $\Int{\pi}{}\cong\Int{\pi'}{}$.
\end{theorem}

We now consider the exponential connectives. We will consider a promotion rule cut against the following rules: dereliction, digging, and contraction. We consider for this a proof $\pi$ of $\pdash A^{\pol},B$, or a proof $\pi$ of $A\pdash ;B$ as both promotion rules are treated similarly, and its interpretation $\Int{\pi}{}\in\cond{A\multimap B}$. Applying a promotion rule (polarised or not) to $\pi$ yields a proof $\rho$ whose interpretation is $\Int{\rho}{}=\oc\Int{\pi}{}\plug\de{prom}$. Then, given a proof $\pi'$ to which we apply one of the three structural rules above to obtain a proof $\rho'$, we consider the interpretations of the proof $\nu$ obtained by a cut between $\pi$ and $\rho$, and the interpretation of the proof $\nu'$ obtained by applying a step of the cut-elimination procedure on $\nu$. It turns out that those are equal, i.e. $\Int{\nu}{}=\Int{\nu'}{}$.

\begin{theorem}\label{onthenose}
If $\pi'$ is obtained from $\pi$ by applying a step of cut-elimination among (promotion/dereliction), (promotion/digging) or (promotion/contraction), then $\Int{\pi}{}=\Int{\pi'}{}$.
\end{theorem}

So execution computes these elimination steps on the nose. What about the last step, namely (promotion/weakening)? In that case, we are faced a problem similar to what happens for additive cuts in \MALL. As execution is a completely local procedure, it cannot erase a whole proof at once. Thus, this elimination step is not soundly represented by the execution. However, one can show the following weaker result.

\begin{theorem}\label{weakcutelim}
If $\pi'$ is obtained from $\pi$ by applying a step of cut-elimination among (promotion/weakening), then $\Int{\pi}{}\cong\Int{\pi'}{}$.
\end{theorem}

However, these results we just showed are not enough to entail that cut-elimination is soundly represented by execution \emph{up to observational equivalence}. In particular, one would need to show that perennisation and observational equivalence interact properly, i.e. one would hope for a result stating that, given two balanced projects $\de{a}$ and $\de{b}$, $\de{a}\cong\de{b}$ if and only if $\de{\oc a}\cong\de{\oc b}$. One can prove that $\de{a}\not\cong\de{b}$ implies $\de{\oc a}\not\cong\de{\oc b}$, however the converse implication is still an open question. We believe that this may be solved by considering a larger notion of equivalence, namely that of \emph{balanced observational equivalence}, i.e. equivalence w.r.t testing by balanced projects. %In terms of categorical models, a positive answer to this question would imply that the functor interpreting exponential connectives is compatible with the congruence defined by (balanced) observational equivalence, i.e. it defines a functor on the quotient category \concat. This result is therefore the last piece missing to show that quotienting the category obtained from our models provide a categorical of linear logic, although the categorical structure of the model may be quite involved due to the polarity constraints.

\section{Quantitative aspects of Interaction Graphs}

One of the key aspects of interaction graphs is their ability to accomodate quantitative aspects. Formally, recent work by the author \cite{seiller-fock} relates the models (for multiplicative linear logic) obtained by the interaction graphs construction to so-called weighted relational models by Laird \emph{et al.} \cite{quantdenot}. It is thus natural to expect Interaction Graphs models of various classical quantitative models of computation such as probabilistic computation. It appears however that the construction is flexible enough to allow for the study of a much larger class of \enquote{quantitative aspects}, as illustrated by two quite different examples that we now discuss: Kennedy's units of measure \cite{unitsofmeasure}, and Ghica and Smith's bounded linear types \cite{Ghica-BLL}.

We will limit the discussions of these examples here either to simply typed lambda-calculus, although the model is more expressive. Detailed study of the models in whole would lead us out of the scope of this paper, and should be the subject of future work.

%\subsection{Interpretation of integers}
%
%We first explain how one can represent integers in the models described above; we chose the unary representation of integers for conciseness, but everything we depict applies to the binary representation as well. The reader should be aware that there is not a unique such representation. In intuitionistic logic, the type $\forall X, (X\Rightarrow X)\Rightarrow(X\Rightarrow X)$, which translate (through the so-called Girard's translation) to $\forall X, \oc( \oc X\multimap X)\multimap \oc X\multimap X$ in linear logic, is used for such a purpose. However, in linear logic, one can also use the type $\forall X, \oc (X\multimap X) \multimap \oc (X\multimap X)$; this latter choice has the advantage to work in constrained linear logic systems such as \ELL. We here follow this second choice, and refer the reader to earlier work for more detailed explanations about integer representation \cite{seiller-conl}. %The important point is that a given integer $n$ will be represented as a  
%
%Let us see how the integer $2$ is represented. It is the interpretation of the proof:
%\begin{prooftree}
%\AxiomC{$\vdash X,X^{\pol}$}
%\AxiomC{$\vdash X,X^{\pol}$}
%\BinaryInfC{$\vdash X\otimes X^{\pol}, X^{\pol}, X$}
%\AxiomC{$\vdash X,X^{\pol}$}
%\BinaryInfC{$\vdash X\otimes X^{\pol}, X\otimes X^{\pol}, X^{\pol}, X$}
%\UnaryInfC{$\vdash X\otimes X^{\pol}, X\otimes X^{\pol}, X^{\pol}\parr X$}
%\UnaryInfC{$\vdash \wn (X\otimes X^{\pol}),\wn (X\otimes X^{\pol}),\oc (X^{\pol}\parr X)$}
%\UnaryInfC{$\vdash \wn (X\otimes X^{\pol}),\oc (X^{\pol}\parr X)$}
%\end{prooftree}

\subsection{Units of measure and bounded linear types}

One of the most straightforward examples is the work by Kennedy \cite{unitsofmeasure} introducing types systems dealing with the notion of \emph{units}. Here, the set of units one wants to deal with, e.g. meters, seconds squared, meters per second, is represented as an (abelian) group $U$.  Then, types are considered dependent upon elements of this group, and one is allowed to consider the type of meters $\Nat[m]$ or the type of functions from masses to energies (such as multiplication by a squared velocity for instance). Moreover one is allowed to consider quantification over types of units. For instance, multiplication of natural numbers could be typed as $\forall u\forall u', \Nat[u]\times\Nat[u']\Rightarrow\Nat[uu']$.

Now, if one considers the Interaction graph model described in this paper in which the monoid $\Omega$ is replaced by the chosen abelian group of units $U$, then we can easily interpret unit-dependent types and quantification over units. A very naive way to do so consists in appending to the interpretation of a term a single edge on a measurable subset used and reserved for the sole purpose of interpreting units, and weight this edge with the appropriate unit. For instance, a natural number of type $\Nat[1]$ (resp. of type $\Nat[m]$) will just be the usual interpretation of the natural number extended by a single edge of weight $1$ (resp. of weight $m$). Moreover, polymorphism is here something we get for free, using the same methods as for second-order quantification: the type $\forall u, \cond{A}[u]$ is simply interpreted as the conduct $\bigcap_{u\in U} \cond{A}[u]$.

%Of course, this naive approach is named well. If one wants to interpret finer quantitative approaches, one needs to be more subtle in defining with the interpretations.

%Now, the more involved type system considered by Ghica 

%Let us describe integers. The type $\Nat[u]$ is simply defined from the interpretation of natural numbers described above where the edge going out of the $\star$ symbol is given the weight $u$. One can then check that the type $\forall u, \Nat[u]$ -- the intersection of all $\Nat[u]$ types -- is empty, as expected, and that usual arithmetic operations can be typed coherently.

%\subsection{Bounded linear types}

Further to this limited setting of Kennedy's unit of measures type system, this same interpretation also provides models of bounded linear types. Bounded linear types were introduced by Ghica and Smith \cite{Ghica-BLL} as a generalisation of Girard, Scedrov and Scott's bounded linear logic \BLL \cite{bll}, a variation of linear logic that accounts for polynomial time computation. The definition of a bounded linear type system is dependent upon a resource semiring $(J,+,\times,0,1)$. It is clear that the underlying monoid $(J,\times,1)$ could be chosen as the weight monoid in our models without hidering the naive interpretation considered above. 

Thus, the only real challenge lies in the interpretation of the sum of the semiring, which is used only in the contraction rule. But one can choose to represent the sum through a splitting of the dialect $[0,1]$. This leads to a correct interpretation of the (exponential) contraction rule since $\de{ctr}$ will map an element of the tensor in $\lambda.\cond{A}\otimes \mu.\cond{A}$ to an element of $(\lambda+\mu).\cond{A}$. In particular, it will take the two single edges encoding the elements of $J$ and superimposing them using the dialect, creating the exact situation we chose for interpreting the sum. 

Another way of thinking about this naive interpretation is that all that is needed to interpret these type systems is a single conduct $\cond{U}$ which possesses subconducts (i.e.\ subtypes) that can be used to simulate algebraic operations in the group $U$ (for units of measure) or in the semiring $J$ (for bounded linear types). Using notations from previous papers, the previous argumentation shows that this conduct $\cond{U}$ can be chosen as $\cond{\oc T_{V}}$ with $V$ a non-negligible measurable set. The naive interpretation $\Int{A}{}$ is then obtained by defining the interpretation of a type $A$ to be its classic interpretation $\Int{A}{c}$ in a simple type system tensored with $\cond{U}$, i.e. $\Int{A}{}=\Int{A}{c}\otimes \cond{U}$.

The generality of our models seems to accomodate for the more general framework allowing for resource variables and resource polymorphism considered by dal Lago and Hoffman \cite{hoffmandallago}.

\subsection{Probabilistic computation}

We now discuss probabilistic computation. We first show how to interpret a simple typed probabilistic lambda calculus without using the monoid of weights, but only the fact that our dialect can be continuous. The probabilistic lambda-calculus we consider \cite{ProbLambda} is defined by the following grammar:
$$ t := x ~|~ \lambda x. t ~|~ tt ~|~ \sum_{i=1}^{k} \alpha_{i} t_{i} $$
where in the last expression the coefficients $\alpha_{i}$ are elements of $[0,1]$ which add up to $1$.

As mentioned in the introduction, future work will show how the models just defined can lead to interpretation of pure lambda-calculus. We are however bound to restrict to typable terms in this paper. We will consider simple types only, although nothing but space constraints prevent us for considering more expressive type system, e.g. system F. We thus consider simple types, and will use the following  typing rule for the additional construct on terms (probabilistic sums):
\begin{prooftree}
\AxiomC{$\Gamma\vdash t_{1}: P$}
\AxiomC{$\dots$}
\AxiomC{$\Gamma\vdash t_{k}: P$}
\TrinaryInfC{$\Gamma\vdash \sum_{i=1}^{k} \alpha_{i} t_{i}: P$}
\end{prooftree}

Now, we will interpret a given term of type $A$ as an element of a conduct $\cond{\oc A}$. The use of the exponential modality will allow us to interpret the probabilistic sum using a variant of the contraction rule. In fact the interpretation can be understood as follows:
\begin{prooftree}
\AxiomC{$\Gamma\vdash t_{1}: P$}
\AxiomC{$\dots$}
\AxiomC{$\Gamma\vdash t_{k}: P$}
\TrinaryInfC{$\Gamma\vdash <t_{1},\dots,t_{k}>: \otimes_{i=1}^{k} P$}
\UnaryInfC{$\Gamma\vdash \sum_{i=1}^{k} \alpha_{i} t_{i}: P$}
\end{prooftree}
where the last operation is a \emph{weighted contraction}. To illustrate how this weighted contraction works, we detail an example with a sum of two terms $\alpha t + (1-\alpha t)$. The weighted contraction, instead of splitting the dialect into $[0,1/2]$ on one hand and $[1/2,1]$ on the other, will split it into $[0,\alpha]$ and $[\alpha,1]$. This graphing can be understood as the realisation of the following thick graph\footnote{Elements of the dialect (resp. support) are listed on a vertical (resp. horizontal) scale; double arrows represent two inverse edges.}:
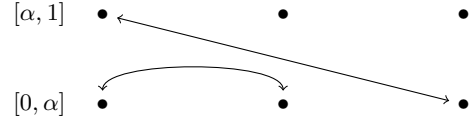
\begin{center}
\begin{tikzpicture}[x=1.2cm,y=0.6cm]]
	\node (A1) at (0,0) {$\bullet$};
	\node (A2) at (0,2) {$\bullet$};
	\node (B1) at (2,0) {$\bullet$};
	\node (B2) at (2,2) {$\bullet$};
	\node (C1) at (4,0) {$\bullet$};
	\node (C2) at (4,2) {$\bullet$};
%	\node (LA) at (0,-0.3) {\scriptsize{{\sc throw}}};
%	\node (LB) at (2,-0.3) {\scriptsize{{\sc head}}};
%	\node (LC) at (4,-0.3) {\scriptsize{{\sc tails}}};
	\node (L1) at (-0.7,0) {$\scriptsize{[0,\alpha]}$};
	\node (L2) at (-0.7,2) {$\scriptsize{[\alpha,1]}$};
	
	\draw[<->] (A1) .. controls (0,1) and (2,1) .. (B1) {};% node [midway, above] {$\frac{1}{2}$};
	\draw[<->] (A2) --  (C1) {}; %node [midway, above] {$\frac{1}{2}$};
	
\end{tikzpicture}
\end{center}

In order to really use the weights though, we will consider the example of a coin-toss operation, which then corresponds simply to the case $\alpha=1/2$ in the above example. By considering the monoid of weights $[0,1]$ with the usual multiplication, we can refine this interpretation by introducing non-trivial weights:
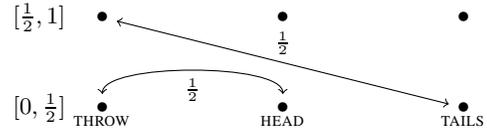
\begin{center}
\begin{tikzpicture}[x=1.2cm,y=0.6cm]]
	\node (A1) at (0,0) {$\bullet$};
	\node (A2) at (0,2) {$\bullet$};
	\node (B1) at (2,0) {$\bullet$};
	\node (B2) at (2,2) {$\bullet$};
	\node (C1) at (4,0) {$\bullet$};
	\node (C2) at (4,2) {$\bullet$};
	\node (LA) at (0,-0.3) {\scriptsize{{\sc throw}}};
	\node (LB) at (2,-0.3) {\scriptsize{{\sc head}}};
	\node (LC) at (4,-0.3) {\scriptsize{{\sc tails}}};
	\node (L1) at (-0.7,0) {$\scriptsize{[0,\frac{1}{2}]}$};
	\node (L2) at (-0.7,2) {$\scriptsize{[\frac{1}{2},1]}$};
	
	\draw[<->] (A1) .. controls (0,1) and (2,1) .. (B1) node [midway, below] {\scriptsize{$\frac{1}{2}$}};
	\draw[<->] (A2) --  (C1) node [midway, above] {\scriptsize{$\frac{1}{2}$}};
	
\end{tikzpicture}
\end{center}

%This refined case is of particular interest when one considers a specific orthogonality defined by the map $m: [0,1]\rightarrow \realposN\cup{\infty}, x\mapsto -\log(1-x)$. Indeed, in this specific case studied in the author's first paper on interaction graphs \cite{seiller-goim}, a graph $G$ is equivalent to its \emph{contraction} $\what{G}$, defined by replacing the set $(e_{i})_{i\in I}$ of all edges in $G$ of source $s$ and target $t$ by a single edge of source $s$ and target $t$ whose weight $\omega$ is the sum of the weights of the edges $e_{i}$: $\omega=\sum_{i\in I} w(e_{i})$.

%Although the model shown above accommodates very expressive type systems, e.g. system F, we will here restrict our discussion to the simply typed lambda-calculus for conciseness.

\section{Conclusion}

Girard's so-called \enquote{GoI3} model \cite{goi3} already provided an interpretation of full linear logic. However, we managed to do so in a quantitative-flavoured framework. As examples, we explained how the models constructed can model various quantitative aspects, such as probabilistic computation, as well as some seemingly type-theoretic constructs such as modalities \emph{\`{a} la} bounded linear logic, or Kennedy's units of measure. These exemplify the wide range of quantitative informations that can be dealt with in the models.

Beyond the results presented here, the adaptation of the interaction graphs framework to deal with continuous dialects results in a more mature and complete construction that opens new directions for future work. In particular, the possibility to restrict or expand the ways map can act on the dialect through the microcosm can be of interest in terms of computational complexity. Indeed, while the weight monoid and the measurement of weights seem to be related to different computational paradigms and can be used, for instance, for representing probabilistic computation, the microcosm can be used to restrict the computational principles allowed in the model and characterise in this way various complexity classes \cite{seiller-towards}. All characterisations considered in the cited work were based on variants of exponential connectives satisfying at least the contraction principle. In the more general construction explained here, we are now able to consider models of exponentials that do not satisfy this principle. In this line of work, it would be interesting to understand if one can adapt Mazza and Terui's work on parsimonious lambda-calculus \cite{MazzaParsimonious,Mazza:CSL2015,MazzaTerui:ICALP2015}, and obtain an interaction graph model for it.

\acks

This work was partly supported by the Marie Sk\l{}odowska-Curie Individual Fellowship (H2020-MSCA-IF-2014) 659920 - ReACT and the ANR 12 JS02 006 01 project \href{http://lipn.univ-paris13.fr/~pagani/pmwiki/pmwiki.php/Coquas/Coquas}{COQUAS}.

\bibliographystyle{abbrvnat}

\end{document}